\documentclass[a4pager,12pt]{article}

\usepackage{latexsym}
\usepackage{amsmath}
\usepackage{amssymb}
\usepackage{color}
\usepackage{epsfig}
\usepackage{theorem}
\usepackage{float}
\usepackage{graphicx}

\restylefloat{table}
\restylefloat{figure}

\newtheorem{theorem}{Theorem}

\newtheorem{observation}{Observation}

\newenvironment{texteqn}
{\begin{equation}\begin{minipage}{0.8\linewidth}}
   {\end{minipage}\end{equation}}

\newenvironment{texteqn*}
{\begin{eqnarray*}\begin{minipage}{0.8\linewidth}}
   {\end{minipage}\end{eqnarray*}}

\newcommand{\eqn}{\begin{texteqn}}
\newcommand{\eofeqn}{\end{texteqn}}

\newcommand{\eqnn}{\begin{texteqn*}}
\newcommand{\eofeqnn}{\end{texteqn*}}

\newcommand{\proof}{{\em Proof.}~}
\newcommand{\qed}{\hfill $\square$}

\title{Two characterizations of chain partitioned probe graphs\thanks{Annals of Operations Research 188 (2011) 279-283. DOI: 10.1007/s10479-010-0749-3}}

\author{ 
  Van Bang Le\\[1ex]
  {\footnotesize
  Universit\"at Rostock, Institut f\"ur Informatik, 
  D-18051 Rostock, Germany.}\\[.5ex]
  {\footnotesize \texttt{le@informatik.uni-rostock.de}}
}
\date{}

\begin{document}
\maketitle

\begin{abstract}
\noindent Chain graphs are exactly bipartite graphs without induced $2K_2$ (a graph with four vertices and two 
disjoint edges). A graph $G=(V,E)$ with a given independent set $S\subseteq V$ (a set of pairwise non-adjacent 
vertices) is said to be a chain partitioned probe graph if $G$ can be extended to a chain graph by adding 
edges between certain vertices in $S$. 
In this note we give two characterizations for chain partitioned probe graphs. The first one describes chain 
partitioned probe graphs by six forbidden subgraphs. The second one characterizes these graphs via a certain 
``enhanced graph'': $G$ is a chain partitioned probe graph if and only if the enhanced graph $G^*$ is a chain graph. 
This is analogous to a result on interval (respectively, chordal, threshold, trivially perfect) partitioned probe 
graphs, and gives an $O(m^2)$-time recognition algorithm for chain partitioned probe graphs.

\end{abstract}


\bigskip
\noindent 
Given a graph class $\mathcal{C}$, a graph $G=(V,E)$ is called a \emph{$\mathcal{C}$ probe graph} if there exists 
an independent set $S\subseteq V$ and a set $E'\subseteq\binom{S}{2}$ such that the graph $G'=(V,E\cup E')$ is in 
the class $\mathcal{C}$, where $\binom{S}{2}$ stands for the set of all 2-element subsets of $S$. 
Probe graphs have been investigated for various graph classes; see~\cite{GolMafMor} and 
the literature given there for more information. 

In case of bipartite graphs without induced $2K_2$ (a graph with four vertices and two disjoint edges), also called 
\emph{chain graphs} by Yannakakis~\cite{Yannakakis} or \emph{difference graphs} by Hammer et al.~\cite{HamPelSun}, 
Golumbic et al.~\cite{GolMafMor} very recently found a characterization of chain probe graphs as follows; 
notice that chain probe graphs are necessarily bipartite.
 
\begin{theorem}[\cite{GolMafMor}]\label{thm:chainprobe}
Let $G$ be a bipartite graph. Then $G$ is a chain probe graph if and only if it is 
$\{C_6,P_7,3K_2,H_1,H_2,H_3\}$-free.
\end{theorem}

Here, given a graph $H$, a graph is $H$-free if it does not have an induced subgraph isomorphic to $H$. 
$C_k$ (respectively, $P_k$) stands for the chordless cycle (respectively, the chordless path) with $k$ vertices, 
and $kK_2$ is the graph with $2k$ vertices and $k$ disjoint edges. See Fig.~\ref{fig:H13} for the graphs $H_1, H_2$ 
and~$H_3$.

\medskip

\begin{figure}[ht]
\begin{center}
\unitlength 0.45mm
\begin{picture}(240,33)
 \put(0,10 ){\line(1,0){60}}
 \put(0,30){\line(1,0){60}}    
 \put(0,10){\line(0,1){20}}
 \put(20,10){\line(0,1){20}}    
 \put(40,10){\line(0,1){20}}     
 \put(60,10){\line(0,1){20}}
    
 \put(0,10 ){\circle*{4}}  
 \put(0,30 ){\circle*{4}} 
 \put(20,10 ){\circle*{4}}  
 \put(20,30 ){\circle*{4}} 
 \put(40,10 ){\circle*{4}}  
 \put(40,30 ){\circle*{4}} 
 \put(60,10 ){\circle*{4}}  
 \put(60,30 ){\circle*{4}} 

 \put(25,-2){\small $H_1$}

 \put(90,10 ){\line(1,0){60}}
 \put(90,30){\line(1,0){60}}    
 \put(90,10){\line(0,1){20}}
 \put(110,10){\line(0,1){20}}    
 \put(130,10){\line(0,1){20}}     
    
 \put(90,10 ){\circle*{4}}  
 \put(90,30 ){\circle*{4}} 
 \put(110,10 ){\circle*{4}}  
 \put(110,30 ){\circle*{4}} 
 \put(130,10 ){\circle*{4}}  
 \put(130,30 ){\circle*{4}} 
 \put(150,10 ){\circle*{4}}  
 \put(150,30 ){\circle*{4}} 

 \put(115,-2){\small $H_2$}

 \put(180,10 ){\line(1,0){60}}
 \put(180,30){\line(1,0){60}}    
 \put(200,10){\line(0,1){20}}    
 \put(220,10){\line(0,1){20}}     
    
 \put(180,10 ){\circle*{4}}  
 \put(180,30){\circle*{4}} 
 \put(200,10 ){\circle*{4}}  
 \put(200,30 ){\circle*{4}} 
 \put(220,10 ){\circle*{4}}  
 \put(220,30 ){\circle*{4}} 
 \put(240,10 ){\circle*{4}}  
 \put(240,30 ){\circle*{4}} 

 \put(205,-2){\small $H_3$}
\end{picture}
\caption{The graphs $H_1, H_2$ and $H_3$}\label{fig:H13} 
\end{center}
\end{figure}
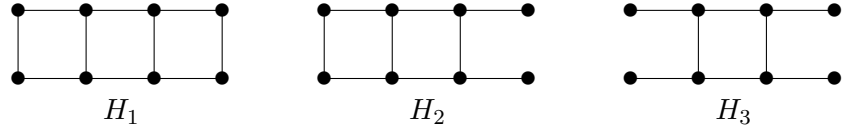

Notice that Theorem~\ref{thm:chainprobe} is good in the sense that it implies a polynomial-time recognition 
algorithm for chain probe graphs. In fact, as pointed out in~\cite{GolMafMor}, an $O(n^2)$-algorithm can be derived 
from the proof of Theorem~\ref{thm:chainprobe}.
 
A graph $G=(V,E)$ with a \emph{given} independent set $S\subseteq V$ is said to be a \emph{$\mathcal{C}$ partitioned 
probe graph} if there exists a set $E'\subseteq\binom{S}{2}$ such that the graph $G'=(V,E\cup E')$ is in the class 
$\mathcal{C}$. Recognizing $\mathcal{C}$ partitioned probe graphs is a special case of the $\mathcal{C}$-\textsc{graph 
sandwich} problem (cf. \cite{GolKapSha}). Notice that a $\mathcal{C}$ probe graph need not be a $\mathcal{C}$ partitioned 
probe graph with respect to the given independent set.

Specifically, Theorem~\ref{thm:chainprobe} does not settle the partitioned case. Our goal is to 
characterize chain partitioned probe graphs. The main results of this note (Theorem~\ref{thm:chainpprobe}) consist of 
\begin{itemize}
 \item a characterization of chain partitioned probe graphs by six forbidden 
   induced subgraphs,
 \item a characterization of chain partitioned probe graphs by a certain ``enhanced'' graph, 
   saying that $G$ is a chain partitioned probe graph if and only if the enhanced graph $G^*$ 
   is a chain graph.
\end{itemize}
While the first result is related to Theorem~\ref{thm:chainprobe} in the unpartitioned case, the second one 
is analogous to some interesting results in the literature: \cite{Zhang} shows that a certain enhanced graph 
of an interval partitioned probe graph must be chordal. \cite{CohGolLipSte} shows that a bipartite graph without 
induced cycles of length $>6$ is chordal-bipartite partitioned probe if a certain enhanced graph is chordal bipartite; 
\cite{GolLip} (see also \cite{CohGolLipSte}) shows that a graph without chordless cycles of length at least 
five is chordal partitioned probe if and only if 
a certain enhanced graph is chordal. \cite{BayLedeR} shows that a graph is a threshold (a trivially perfect) 
partitioned probe graph if and only if a certain enhanced graph is a threshold (a trivially perfect) graph;  
see~\cite{BraLeSpi,Golumbic} for graph classes not defined here.

We remark that our characterizations are good in the sense that they imply polynomial-time recognition algorithms 
for chain partitioned probe graphs. 

\subsection*{Characterizing chain partitioned probe graphs}
Let $G=(V,E)$ be a graph with a given independent set $S\subseteq V$. Suppose there exists a set 
$E'\subseteq \binom{S}{2}$ such that the graph $G'=(V,E\cup E')$ is a chain graph, that is, $G$ is a chain 
partitioned probe graph with respect to the given independent set $S$. Then, clearly, 
any induced $2K_2$ in $G$ must have exactly two vertices $x,y$ in $S$ and $\{x,y\}$ must belong to $E'$. 

In what follows, given a graph $G=(V,E)$ together with an independent set $S\subseteq V$, the \emph{enhanced graph} 
$G^*=(V,E^*)$ is obtained from $G$ by adding all edges between two vertices in $S$ that are two vertices 
of an induced $2K_2$ in $G$. 

Chain partitioned probe graphs can be characterized as follows; see Fig.~\ref{fig:F16} for the graphs 
$F_1,\ldots, F_6$.
\begin{theorem}\label{thm:chainpprobe}
Let $G=(V,E)$ be a bipartite graph, and let $S\subseteq V$ be an independent set of vertices of $G$. 
Then the following statements are equivalent.
\begin{itemize}
 \item[\em (i)] $G$ is a chain partitioned probe graph with respect to $S$;
 \item[\em (ii)] $G$ is $\{F_1, F_2, F_3, F_4, F_5, F_6\}$-free;
 \item[\em (iii)] $G^*$ is a chain graph.
\end{itemize}
\end{theorem}

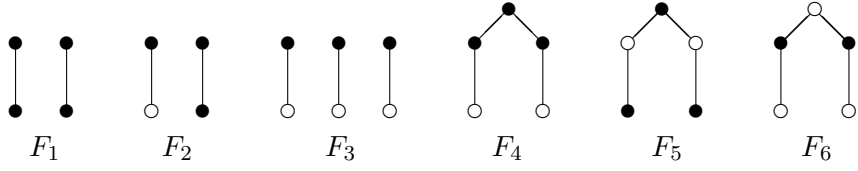
\begin{figure}[ht]
\begin{center}
\unitlength 0.45mm
\begin{picture}(245,43)
 \put(0,10){\line(0,1){20}}
 \put(15,10){\line(0,1){20}}    
    
 \put(0,10 ){\circle*{4}}  
 \put(0,30 ){\circle*{4}} 
 \put(15,10 ){\circle*{4}}  
 \put(15,30 ){\circle*{4}}  

 \put(4,-3){\small $F_1$}

 \put(40,12){\line(0,1){18}}
 \put(55,10){\line(0,1){20}}    
    
 \put(40,10 ){\circle{4}}  
 \put(40,30 ){\circle*{4}} 
 \put(55,10 ){\circle*{4}}  
 \put(55,30 ){\circle*{4}}  

 \put(43,-3){\small $F_2$}

 \put(80,12){\line(0,1){18}}
 \put(95,12){\line(0,1){18}}    
 \put(110,12){\line(0,1){18}}    
 
 \put(80,10 ){\circle{4}}  
 \put(80,30 ){\circle*{4}} 
 \put(95,10 ){\circle{4}}  
 \put(95,30 ){\circle*{4}}  
 \put(110,10 ){\circle{4}}  
 \put(110,30 ){\circle*{4}}  

 \put(91,-3){\small $F_3$}

 \put(135,12){\line(0,1){18}}
 \put(155,12){\line(0,1){18}}    
 \put(135,30){\line(1,1){11}}
 \put(155,30){\line(-1,1){11}}
  
 \put(135,10 ){\circle{4}}  
 \put(135,30 ){\circle*{4}} 
 \put(155,10 ){\circle{4}}  
 \put(155,30 ){\circle*{4}}  
 \put(145,40){\circle*{4}}

 \put(140,-3){\small $F_4$}

 \put(180,10){\line(0,1){18}}
 \put(200,10){\line(0,1){18}}    
 \put(181.5,31.5){\line(1,1){8.5}}
 \put(198.5,31.5){\line(-1,1){8.5}}
  
 \put(180,10 ){\circle*{4}}  
 \put(180,30 ){\circle{4}} 
 \put(200,10 ){\circle*{4}}  
 \put(200,30 ){\circle{4}}  
 \put(190,40){\circle*{4}}

 \put(187,-3){\small $F_5$}

 \put(225,12){\line(0,1){18}}
 \put(245,12){\line(0,1){18}}    
 \put(225,30){\line(1,1){8.6}}
 \put(245,30){\line(-1,1){8.6}}
  
 \put(225,10 ){\circle{4}}  
 \put(225,30 ){\circle*{4}} 
 \put(245,10 ){\circle{4}}  
 \put(245,30 ){\circle*{4}}  
 \put(235,40){\circle{4}}

 \put(231,-3){\small $F_6$}
\end{picture}
\caption{Forbidden induced subgraphs for chain partitioned probe graphs; white vertices in $S$, black vertices 
in $V\setminus S$}\label{fig:F16} 
\end{center}
\end{figure}
\proof 
By inspection, it can be seen that the graphs $F_1,\ldots, F_6$ depicted in Fig.~\ref{fig:F16} are not 
chain partitioned probe graphs with respect to the black/white partition. Hence the implication (i) $\Rightarrow$ 
(ii) holds. The implication (iii) $\Rightarrow$ (i) is obvious. 

We now are going to show (ii) $\Rightarrow$ (iii). We write $x\sim y$ and $x\sim^* y$ ($x\not\sim y$ and 
$x\not\sim^* y$), meaning vertex $x$ is adjacent (non-adjacent) to vertex $y$ in $G$, respectively, in $G^*$. 
$N_G(x)$ (respectively, $N_{G^*}(x)$) denotes the neighborhood of $x$ in $G$ (respectively, in $G^*$). 
For $U\subseteq V$, $G[U]$ (respectively, $G^*[U]$) is the subgraph of $G$ (respectively, of $G^*$), induced by $U$. 

Assuming (ii), the following observation will simplify some arguments in proving (iii).

\begin{observation}\label{obs}
Suppose $x, y\in S$, $z\in V$ such that $x\sim^*y$ and $z\not\sim^* x,y$. 
Let $x_y\in N_G(x)$, $y_x\in N_G(y)$ such that $G[x,x_y, y,y_x]$ is a $2K_2$ in $G$. 
Then
\begin{itemize}
 \item[\em (1)] $z\not\sim x_y, y_x$;
 \item[\em (2)] If $z\in S$, then, for all $v\in N_G(z)$, $v\sim x, y_x$ and 
    $v\not\sim y, x_y$, or $v\not\sim x, y_x$ and $v\sim y, x_y$.
\end{itemize} 
\end{observation}
\emph{Proof of Observation} \ref{obs}. 

\noindent (1): If $z$ is adjacent to both $x_y$ and $y_x$, then $G[x,x_y,z,y_x,y]$ is an $F_4$ (if $z\not\in S$) 
or an $F_6$ (otherwise). If $z$ is adjacent to exactly one of $x_y$ and $y_x$, say $x_y$, then $G[y,y_x,z,x_y]$ 
is an $F_2$ (if $z\not\in S$) or a $2K_2$ with $y,z\in S$. The latter implies that $y$ and $z$ must 
be adjacent in $G^*$. In any case we reach a contradiction, hence (1).\\[1ex]
\noindent (2): As $G$ is bipartite and $z$ is non-adjacent to $x,y$ in $G^*$, $v$ must be adjacent in $G$ to 
exactly one of $x, x_y$ and to exactly one of $y,y_x$. If $v$ is adjacent to $x$ and $y$ or to $x_y$ and $y_x$, then 
$G[x,y,x_y,y_x,v]$ is an $F_5$ or an $F_4$. Hence (2) holds.\qed

\medskip
We now show that $G^*$ is a chain graph by proving Claims 1 -- 3 below.

\medskip
\noindent\textbf{Claim 1.} $G^*$ is $2K_2$-free.\\
\emph{Proof of Claim} 1. Let, by contradiction, $ab, cd$ be two edges of an induced $2K_2$ in $G^*$. 
Since $G$ has no $F_1$ and no $F_2$, $a,b\in S$ or $c,d\in S$. By symmetry we may assume that $a,b\in S$. 
Let $a_b\in N_G(a)$, $b_a\in N_G(b)$ such that $G[a,a_b,b,b_a]$ is a $2K_2$. 
By Observation~\ref{obs} (1), 
\begin{equation*}
 c\not\sim a_b, b_a\, \text{ and }\, d\not\sim a_b, b_a.
\end{equation*}
Hence, 
\begin{equation*}
 c, d\in S,
\end{equation*}
otherwise, $G[a,a_b,c,d]$ would be an $F_2$ (if $c,d\not\in S$) or $a, c$ would be adjacent in $G^*$ 
(if $c\in S$, $d\not\in S$) or $a, d$ would be adjacent in $G^*$ (if $c\not\in S$, $d\in S$). 

Let $c_d\in N_G(c)$, $d_c\in N_G(d)$ such that $G[c,c_d,d,d_c]$ is a $2K_2$. 
By Observation~\ref{obs} (1) again, 
\begin{equation*}
 a\not\sim c_d, d_c\, \text{ and }\, b\not\sim c_d, d_c,
\end{equation*}
and we reach a contradiction to Observation~\ref{obs} (2) with $x=a, y=b, z=c$, and $v=c_d$. This proves Claim 1.

\medskip
\noindent\textbf{Claim 2.} $G^*$ is $C_5$-free.\\
\emph{Proof of Claim} 2. Let, by contradiction, $ab, bc, cd, de, ea$ be the edges of an induced $C_5$ in $G^*$. 
Since $G$ has no $C_5$, at least one of these edges must connect two vertices in $S$. By symmetry we may assume 
that 
\begin{equation*}
 a,b\in S.
\end{equation*} 
Furthermore, since $G$ has no $F_4$ and $F_6$, $c$ or $e$ must belong to $S$. By symmetry we 
may assume that 
\begin{equation*}
 c\in S. 
\end{equation*}
Let $a_b\in N_G(a)$, $b_a,b_c\in N_G(b)$, $c_b\in N_G(c)$ such that 
each of $G[a,a_b,b,b_a]$ and $G[b,b_c,c,c_b]$ is a $2K_2$. Since $G$ is bipartite, 
\begin{equation*}
 b_a\not\sim b_c.
\end{equation*}
Moreover, by Observation~\ref{obs} (1), 
\begin{equation*}
 d\not\sim b_a,\, e\not\sim b_c. 
\end{equation*}
Furthermore,
\begin{equation*}
 d\sim b_c,\, e\sim b_a,
\end{equation*}
otherwise $G^*[d,e,b,b_c]$ or $G^*[d,e,b,b_a]$ would be a $2K_2$, contradicting Claim~1. (Notice that, as 
$b_c, b_a\not\in S$, $d\not\sim b_c$, respectively, $e\not\sim b_a$ implies $d\not\sim^* b_c$, respectively, 
$e\not\sim^* b_a$)

In particular, $b_a\not= b_c$. Now, if $d\not\in S$ or $e\not\in S$, then $G[d,e,b_a,b,b_c]$ is a $C_5$. 
If both $d$ and $e$ belong to $S$, then $G[d,b_c,b,b_a,e]$ is an $F_6$. In any case we reach a contradiction, 
and Claim~2 follows.

\medskip
\noindent\textbf{Claim 3.} $G^*$ is $C_3$-free.\\
\emph{Proof of Claim} 3. Let, by contradiction, $a,b,c$ be the vertices of a triangle in $G^*$. 
Since $G$ is bipartite, at least two of $a,b,c$ must belong to $S$. By symmetry we may assume that 
$a, b\in S$. Let $a_b\in N_G(a)$, $b_a\in N_G(b)$ such that $G[a,a_b,b,b_a]$ is a $2K_2$. Now if 
$c\not\in S$, then $G[a_b,a,c,b,b_a]$ is an $F_5$, a contradiction.

So, let $c\in S$, and let $a_c\in N_G(a)$, $b_c\in N_G(b)$, $c_a, c_b\in N_G(c)$ such that 
each of $G[b,b_c,c,c_b]$ and $G[a,a_c,c,c_a]$ is a $2K_2$. Note that possibly 
$a_b=a_c, b_a=b_c$ and $c_a=c_b$. However, as $G$ is bipartite, $a_b\not\sim a_c$, $b_a\not\sim b_c$ and 
$c_a\not\sim c_b$. 
Moreover, 
\begin{equation*}
 c\not\sim a_b,
\end{equation*}
otherwise $a_b\not= a_c$ and $a_c, a, a_b, c, c_a$ would induce an $F_5$ (if $c_a\not\sim a_b$) or $c_a,a_b,c$ 
would induce a triangle in $G$ (if $c_a\sim a_b$). 
By symmetry,
\begin{equation*}
 c\not\sim b_a,\, a\not\sim b_c,\, a\not\sim c_b,\, b\not\sim a_c,\, b\not\sim c_a.
\end{equation*}   
In particular, $\{a_b, a_c\}$, $\{b_a, b_c\}$ and $\{c_a,c_b\}$ are pairwise disjoint. 
Furthermore,
\begin{equation*}
 c_a\not\sim a_b, b_a,
\end{equation*}
otherwise, say $c_a\sim a_b$, then $c_a,a_b,b,b_a$ would induce an $F_2$ (if $c_a\not\sim b_a$) or 
$a, a_b, c_a, b_a, b$ would induce an $F_4$ in $G$ (if $c_a\sim b_a$).

But then $a, a_b, b, b_a, c$ and $c_a$ induce an $F_3$. This contradiction completes the proof of Claim 3.

\medskip
It follows from Claims 1 -- 3 that $G^*$ is a bipartite graph without induced $2K_2$, that is, (ii) $\Rightarrow$ (iii) 
holds. The proof of Theorem~\ref{thm:chainpprobe} is complete.\qed

\medskip
To conclude the note we remark that the characterization (i) $\Leftrightarrow$ (iii) stated in 
Theorem~\ref{thm:chainpprobe} easily implies an $O(m^2)$-algorithm to recognize chain partitioned probe graphs, 
where $m$ is the number of edges: 
Given a (bipartite) graph $G=(V,E)$ with an independent set $S$, the enhanced graph $G^*=(V,E\cup E')$ 
can be constructed in time $O(m^2)$ by computing all induced $2K_2$ in $G$. Then checking if $G^*$ is a chain graph, 
i.e, bipartite and $2K_2$-free, obviously can be done in the same time complexity. 

\bigskip
The author would like to thank Fr\'ed\'eric Maffray and two anonymous referees for their helpful 
comments and suggestions.


\end{document}